\documentclass[aps,pra,twocolumn,superscriptaddress,nofootinbib]{revtex4-1}
\usepackage{epsfig}
\usepackage{amsmath,amssymb,amsfonts}
\usepackage{amsthm}
\theoremstyle{remark}
\usepackage{graphicx}
\newcommand{\beq}{\begin{equation}}
\newcommand{\eeq}{\end{equation}}
\newcommand{\nn}{\nonumber}

\usepackage{braket}
\newtheorem{theorem}{Theorem}

\newtheorem{proposition}[theorem]{Proposition}
\newtheorem{definition}{Definition}

\newcommand{\ba}[2]{\begin{align}\label{#1}#2\end{align}}
\usepackage{mathtools}

\begin{document}

\title{Maximal tree size of few-qubit states}
\author{Huy Nguyen \surname{Le}}
\affiliation{Centre for Quantum Technologies, National University of Singapore, 3 Science Drive 2, Singapore 117543, Singapore}
\author{Yu \surname{Cai}}
\affiliation{Centre for Quantum Technologies, National University of Singapore, 3 Science Drive 2, Singapore 117543, Singapore}
\author{Xingyao \surname{Wu}}
\affiliation{Centre for Quantum Technologies, National University of Singapore, 3 Science Drive 2, Singapore 117543, Singapore}
\author{Rafael \surname{Rabelo}}
\affiliation{Departamento de Matem\'atica, Universidade Federal de Minas Gerais, Caixa Postal 702, Belo Horizonte, Minas Gerais 30123-970, Brazil}
\author{Valerio \surname{Scarani}}
\affiliation{Centre for Quantum Technologies, National University of Singapore, 3 Science Drive 2, Singapore 117543, Singapore}
\affiliation{Department of Physics, National University of Singapore, 2 Science Drive 3, Singapore 117542, Singapore}

\begin{abstract}
Tree size ($\rm{TS}$) is an interesting measure of complexity for multiqubit states: not only is it in principle computable, but one can obtain lower bounds for it. In this way, it has been possible to identify families of states whose complexity scales superpolynomially in the number of qubits. With the goal of progressing in the systematic study of the mathematical property of $\rm{TS}$, in this work we characterize the tree size of pure states for the case where the number of qubits is small, namely, 3 or 4. The study of three qubits does not hold great surprises, insofar as the structure of entanglement is rather simple; the maximal $\rm{TS}$ is found to be 8, reached for instance by the $\ket{\rm{W}}$ state. The study of four qubits yields several insights: in particular, the most economic description of a state is found not to be recursive. The maximal $\rm{TS}$ is found to be 16, reached for instance by a state called $\ket{\Psi^{(4)}}$ which was already discussed in the context of four-photon down-conversion experiments. We also find that the states with maximal tree size form a set of zero measure: a smoothed version of tree size over a neighborhood of a state ($\epsilon-\rm{TS}$) reduces the maximal values to 6 and 14, respectively. Finally, we introduce a notion of tree size for mixed states and discuss it for a one-parameter family of states.

\end{abstract}

\begin{widetext}
\maketitle
\end{widetext}

\section{Introduction} 
It is likely that the origin of the ``speed up" quantum computers offer over classical computers lie in the quantum states. There are many evidences that if the state in a quantum computation is simple, it can be simulated efficiently with classical computers. Examples are quantum circuits where the states at every step have polynomial Schmidt rank \cite{Vidal03}, and measurement-based quantum computation on resource states with logarithmically bounded Schmidt-rank width \cite{VDN07}, or polynomial tree size \cite{us13}. States that are useful for quantum computing must also be realizable by a quantum circuit of polynomial size. So, for quantum computation to have an advantage over its classical counterpart, a state must be sufficiently complex in some sense but can be prepared efficiently. Thus, studying the complexity of states is important because it is not only a fundamental concept of nature but also a relevant aspect in quantum computing.

Among the different measures of complexity, quantum Kolmogorov complexity is the attempt to quantify complexity of states in the most general way \cite{Mora05,Mora07,Rogers08}; however, this measure suffers from the setback that it is not computable and only upper bounds can be given. Therefore it is possible to certify that a state is not complex, but it is impossible to certify that a state is complex. Things are better defined when we restrict consideration to some particular representations. If we restrict our attention to the most common description of quantum states, Dirac's bra-ket notation, it is possible to prove superpolynomial lower bounds \cite{Aaronson04,Raz04}. The cost of expressing a state with bra-ket notation gives rise to the definition of tree-size complexity. A state with a very long bra-ket representation is hard to generate with classical computers, so tree size is a ``classical complexity" measure as it indicates the difficulty of simulating a state using classical means. We must stress that bra-ket notation is not the only possible classical description of states; another well known one is the matrix product  representation (MPS) in which the cost is associated with the size of the matrices \cite{Perez07}. In Ref.~\cite{us13}, we obtained a relation between tree size and the size of the matrices in a MPS, which shows that tree size is only a polynomial in the number of qubits when the matrix size in a MPS is bounded.     

While lower bounds on tree size can be obtained by utilizing Raz's theorem on multilinear formulas \cite{Aaronson04, Raz04}, finding the exact tree size of a given state remains to be investigated. In this paper, we use an exhaustive procedure to compute the tree size and find the \emph{most complex state} for three and four qubits. From now on when we mention the ``most complex" state we mean the state with \emph{maximal tree size}. Our approach relies on entanglement classification by stochastic local operation and classical communication (SLOCC) for three and four qubits \cite{Dur00,Lamata06,Lamata07}. 

We fist give a brief description of tree-size complexity in Sec.~\ref{TSC} before reviewing relevant results in entanglement classification in Sec.~\ref{entclas}. We then show in Sec.~\ref{3qubit} that the most complex class of three-qubit states is the W class with tree size $8$. However, the tree size of these states is not ``stable" in the sense that an arbitrarily small perturbation in the states leads to a decrease in the tree size to 6. The $\rm{TS}$ of mixed states, particularly that of the generalized Werner states, is also considered. The case of four qubits is discussed in Sec.~\ref{4qubit}: We find that the most complex class is an entanglement class which had been overlooked in the previous works on inductive classification of entanglement \cite{Lamata07,Lamata06}. One example of this class is the state
\small
\ba{}{\ket{\Psi^{(4)}}=\sqrt{\frac{1}{3}}\bigg[&\frac{1}{2}\left(\ket{0110}+\ket{0101}+\ket{1001}+\ket{1010}\right)\nn \\
&-\ket{0011}-\ket{1100}\bigg].}
\normalsize

\section{tree-size complexity}\label{TSC}

The most obvious way to write down a state is using bra-ket notation. Any multiqubit state written in its bra-ket form can be described by a rooted tree of $\otimes$ and $+$ gates; each leaf vertex is labeled with a single-qubit superposition $\alpha \ket{0}+\beta\ket{1}$ (this state need not be normalized) \cite{Aaronson04}. For example, the Bell state $(\ket{00}+\ket{11})/\sqrt{2}$ can be described by the rooted tree $\rm{T_E}$ of entangled two-qubit states in Fig.~\ref{TE} by assigning appropriate single qubit states to each leaf. The \textit{size} of a rooted tree is defined as the number of leaves. A quantum state may be represented by different rooted trees each with a different size. For example, the  state  $\left(\ket{00}+\ket{01}+\ket{10}+\ket{11}\right)/2$ whose size is 8 can also be written as $\ket{+}\ket{+}$ with size 2. The tree size of a state is taken as the \emph{minimum size over all possible trees}. It can be understood as the length of the shortest bra-ket representation of a state.
\begin{figure}[t]
\centering
\includegraphics[scale=0.75]{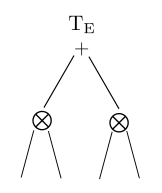}
  \caption{Rooted trees of two-qubit entangled states.}\label{TE}
\end{figure}

Going to the example of three qubits, any pure state can be written as \cite{Acin00}
\ba{acin}{
\ket{\Psi}=\cos \theta \ket{000}+\sin \theta \ket{1} \left(\cos \omega \ket{0'0''}+\sin \omega \ket{1'1''}\right),
}
where the prime and double prime indicate different bases. Thus, tree size is at most 8 for three qubits. We see below that there are indeed states with tree size 8, that is, these states do not have a simpler decomposition.

A more physical measure of complexity which allows for deviations over a neighborhood is the $\epsilon$-approximate tree size of a state. Given a positive  $\epsilon <1$, the $\epsilon$-approximate tree size $\mathrm{TS}_{\epsilon}\left(\ket{\Psi}\right)$ of the state $\ket{\Psi}$ is the minimum tree size over all states $\ket{\varphi}$ such that $|\braket{\Psi|\varphi}|^2\geq 1-\epsilon$ \cite{Aaronson04}. Under a distance error of $\epsilon$, the state $\ket{\Psi}$ is not distinguishable from $\ket{\varphi}$ and hence can be approximated by the latter.

An important property of tree size which we use extensively in this work is that it is invariant under SLOCC. More specifically, we have the following propositions:
\begin{proposition}\label{SLOCCTS}\cite{us13}
If there exist invertible local operators (ILOs) $A_{i}$ such that 
\beq
\ket{\psi}=A_1\otimes \dots \otimes A_n \ket{\phi},
\eeq
then $\mathrm{TS}\left(\ket{\psi}\right)=\mathrm{TS}\left(\ket{\phi}\right)$.
\end{proposition}
The reader may refer to Ref.~\cite{us13} for a proof of this proposition. This is a very useful observation because if we know the tree size of a given state, we know the tree size of all states in its class. Also, entanglement classification by SLOCC has been studied with great detail in the literature.

For a given number of qubits $n$ and a size $\rm{S}$, the number of trees with size at most  $\rm{S}$ is finite. For example, all the trees of three qubits with at most eight leaves are listed in Fig.~\ref{F38} according to their depth. A reader who is familiar with entanglement classification may immediately recognize that these trees correspond to the product, biseparable, Greenberger-Horne-Zeilinger (GHZ), and W families of states, respectively.

Our procedure to find the most complex pure states is as follows: First, we find a decomposition that can describe a particular set of $n$-qubit states and denote the size of this decomposition $\rm{S}$. Next, we show that these states cannot be described by the trees with size smaller than $\rm{S}$, but the other states that are not in this particular set can be. It follows that the maximal tree size is $\rm{S}$ and the set of states mentioned above are the most complex states. 

Tree size can also be extended to mixed states. A mixed state can be decomposed into an ensemble of pure states $\ket{\psi_i}$ as
\ba{}{\rho=\sum_i p_i \ket{\psi_i}\bra{\psi_i}.}
Following the approach of Ref.~\cite{Terhal00}, we define the $\rm{TS}$ of mixed states as
\ba{}{\rm{TS}(\rho)= \min\{\max_{i}[\rm{TS}(\psi_i)]\},}
where the minimization is taken over all possible decompositions of $\rho$. The $\rm{TS}$ of a mixed states is equal to the $\rm{TS}$ of the most complex pure state in its decomposition. The motivation behind this definition is that, if in a preparation procedure a pure state $\ket{\psi_i}$ is realized with nonzero probability, then the prepared state is at least as complex as this pure state.
\begin{figure*}
\centering
\includegraphics[scale=0.45]{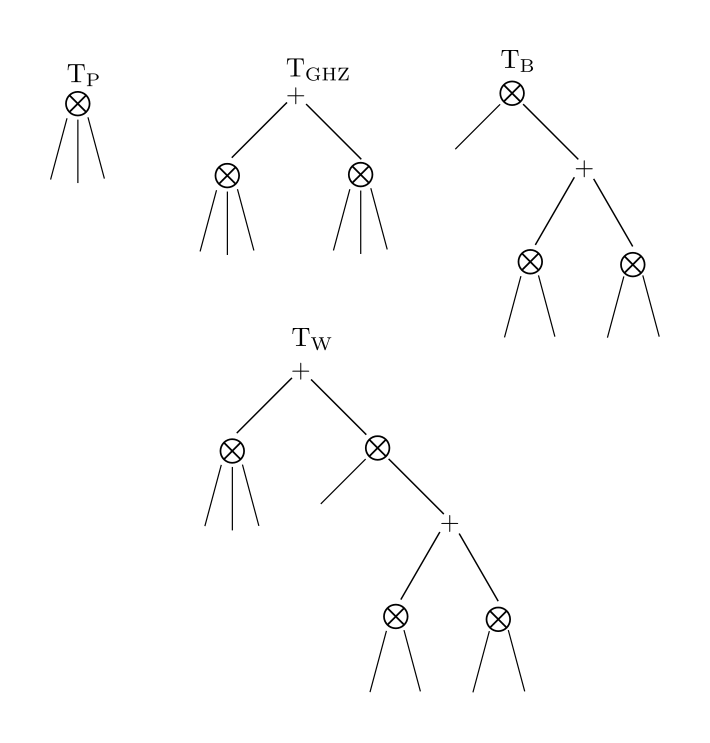}
  \caption{Rooted trees  with size at most 8 for three qubits.}\label{F38}
\end{figure*}

\section{Three qubits}\label{threequbit}
\subsection{Tool: SLOCC classification of three qubits}\label{entclas} 
Given the central role of Proposition~\ref{SLOCCTS} in this work, we start by reviewing the SLOCC entanglement classification of three qubits. These are mostly known results, but we recast them in a form useful for this work. 

The inequivalent classes of a state can be inferred from its coefficient matrix. Let us first start with a two-qubit state
\beq
\ket{\Psi}=c_{00}\ket{00}+c_{01}\ket{01}+c_{10}\ket{10}+c_{11}\ket{11}
\eeq
with the coefficient matrix 
\ba{}{
C=\begin{pmatrix}
c_{00} &c_{01} \nn \\
c_{10} &c_{11} \nn
\end{pmatrix}.
}
It is straightforward to check that this state is a product state if and only if (iff) $\det (C)=0$, that is, $c_{00}c_{11}=c_{01}c_{10}$. And it is entangled iff  $\det (C)\neq 0$. These are the only two inequivalent SLOCC classes of two qubits.

For the case of three qubits, there are six different classes: product, biseparable, GHZ and W classes \cite{Dur00,Lamata06}. Permutation of qubits lead to three inequivalent biseparable classes. The most common examples of states in these classes are
\ba{pbgw}{
\ket{\rm{P}}=&\ket{000}, \nn \\
\ket{\rm{B}}=&\frac{1}{\sqrt{2}}\ket{0}\left(\ket{01}+\ket{10}\right), \nn \\
\ket{\rm{GHZ}}=&\frac{1}{\sqrt{2}}\ket{000}+\ket{111}, \nn \\
\ket{\rm{W}}=&\frac{1}{\sqrt{3}}\left(\ket{001}+\ket{010}+\ket{100}\right).}
A state $\ket{\Psi}$ is said to be in the W class, for example, if and only if there exist invertible local operators $A_{1},A_2,A_3$ such that $\ket{\Psi}=A_1\otimes A_2\otimes A_3 \ket{\rm{W}}$. The product, biseparable, and GHZ classes are defined similarly.

Again, one can determine which class a state belongs to by studying its coefficient matrix. The state is first expressed in a basis expansion
\ba{3qubitstate}{
\ket{\Psi}&= c_{000}\ket{000}+c_{001}\ket{001}+c_{010}\ket{010}+c_{011}\ket{011}\nn \\
&+ c_{100}\ket{100}+c_{101}\ket{101}+c_{110}\ket{110}+c_{111}\ket{111}\nn \\
&= \ket{0}\ket{\phi_0}+\ket{1}\ket{\phi_1},
}
where the two-qubit states are
\ba{phi0phi1}{
\ket{\phi_0}&=c_{000}\ket{00}+c_{001}\ket{01}+c_{010}\ket{10}+c_{011}\ket{11}, \nn \\
\ket{\phi_1}&=c_{100}\ket{00}+c_{101}\ket{01}+c_{110}\ket{10}+c_{111}\ket{11}
}
with the coefficient matrices
\ba{c0c1}{
C^{1|23}_0=\begin{pmatrix}
c_{000} &c_{001} \\
c_{010} &c_{011} 
\end{pmatrix}, \ \
C^{1|23}_1=\begin{pmatrix}
c_{100} &c_{101}  \\
c_{110} &c_{111} 
\end{pmatrix},
}
where the superscript $1|23$ indicates the partition of qubits. Since it will be clear from the context which partition is considered, this superscript will be dropped from the text below. In this work we are more concerned with the states in the GHZ and W classes because they have larger tree size. We now state the conditions that the matrices $C_0$ and $C_1$ must satisfy for $\ket{\Psi}$ to be in the GHZ or W class, which is derived by Lamata \emph{et al.} in Ref.~\cite{Lamata06}. The theorem is stated in a slightly modified but equivalent form which we think is easier to work with.
\begin{proposition}\label{Lamata}\cite{Lamata06}
Let $\ket{\Psi}$ be a three-qubit pure state,  then 

(1) $\ket{\Psi}$ is a GHZ state iff one of the following conditions holds:

\indent \indent (a) There is a partition $i|jk$ for which $C_0$ and $C_1$ are linearly independent, $\det (C_0)\neq 0$, and $C_0^{-1}C_1$ has two distinct eigenvalues.

\indent \indent (b) The same as (a) but with $C_0$ and $C_1$ interchanged. 

\indent \indent (c) For \emph{all} partitions $i|jk$ $C_0$ and $C_1$ are linearly independent, and there is a partition such that $\det (C_0)=\det (C_1)=0$. 

(2) $\ket{\Psi}$ is a W state iff one of the following conditions holds:

\indent \indent (a) There is a partition $i|jk$ for which $C_0$ and $C_1$ are linearly independent,
 $\det (C_0)\neq 0$, and $C_0^{-1}C_1$ has only one eigenvalue. 

\indent \indent (b) The same as (a) but with $C_0$ and $C_1$ interchanged.
\end{proposition}
Note that the eigenvalue equation of a $2\times 2$ matrix  $\rm{M}$ is
\ba{}{
\lambda^2-\rm{tr} (M)\lambda+\det (M)=0;}
hence it has only one eigenvalue when the discriminant vanishes, which yields  $\left[\rm{tr}(M)\right]^2= 4\det(\rm{M})$; otherwise it has two distinct eigenvalues.

A generic state of three qubits almost always obey (1a) or (1b) of Proposition \ref{Lamata}. The other constraints are equations that only a specific set of matrix coefficients satisfy. Thus, most of the states in the Hilbert space of three qubits are in the GHZ class. 

The states in the W class have a special property that is important for finding their $\epsilon$-approximate tree size: For every state in the W class, there is a GHZ state that is arbitrarily close to it \cite{Eisert06}. It is straightforward to verify that the coefficient matrix of the $\ket{\rm{W}}$ state satisfies  condition (2a) of Proposition \ref{Lamata}. But if we introduce a small fluctuation  
\ba{}{
\ket{\rm{W}}\rightarrow \ket{\rm{W}}+\mu \ket{111},
} 
then the coefficient matrix satisfies condition (1a) of Proposition \ref{Lamata} and hence the perturbed state belongs to the GHZ class. Not only $\ket{\rm{W}}$ but every state in its class is ``unstable" under an arbitrarily small fluctuation. For a state $\ket{\phi_{\rm{W}}}$ in this class, there exist invertible local operators $A_1,A_2,A_3$ such that
\ba{}{
\ket{\phi_{\rm{W}}}=A_1\otimes A_2 \otimes A_3 \ket{\rm{W}}.}
We have
\ba{}{
&\ket{\phi_{\rm{W}}}+\mu A_1\otimes A_2 \otimes A_3 \ket{111}\nn \\&=A_1\otimes A_2\otimes A_3\left(\ket{\rm{W}}+\mu \ket{111}\right),}
which is a state in the GHZ class for arbitrarily small $\mu$. 

\subsection{Maximal tree size}\label{3qubit}
It follows from Eq.~\eqref{acin} that the tree size of a three-qubit state cannot exceed $8$. It turns out that all the states in the W class have tree size of exactly $8$.
\begin{proposition}
The most complex states of three qubits are the states in the W class, and they have tree size 8. 
\end{proposition}
\begin{proof}
This is done by ruling out the smaller trees in Fig.~\ref{F38} as possible representations of the $\ket{\mathrm{W}}$ state. First, we observe that the tree $\mathrm{T_P}$ is a special case of $\rm{T_B}$. Indeed, if one branch of the $+$ gate in $\rm{T_B}$ vanishes we get $\rm{T_P}$. Similarly, $\mathrm{T_B}$ is a special case of $\rm{T_{GHZ}}$ because if we expand the $+$ gate in $\mathrm{T_B}$ we obtain the form of $\rm{T_{GHZ}}$. Let us denote by $\ket{\rm{T}}$ the set of states describable by the tree $\rm{T}$, so $\ket{\rm{T_{GHZ}}}$ can be a GHZ, biseparable or product state. We can parametrize $\ket{\rm{T_{GHZ}}}$ as
\ba{}{
\ket{\rm{T_{GHZ}}}=\bigotimes_{i=1}^{3}\left(x_i\ket{0}+y_i\ket{1}\right)+\bigotimes_{i=1}^{3}\left(x'_i\ket{0}+y'_i\ket{1}\right)}
with complex coefficients $x_i,y_i,x'_i$ and $y'_i$. If $\ket{\rm{W}}=\ket{\rm{T_{GHZ}}}$, by equating coefficients of both sides we obtain a system of equations for the variables $x_i,y_i,x'_i,y'_i$. These equations can be easily shown to have no solution. Therefore, the tree $\rm{T_{GHZ}}$ cannot describe the $\ket{\rm{W}}$ state. Now that all the smaller trees have been ruled out, we conclude that the $\ket{\mathrm{W}}$ state has tree size $8$ which is maximal for three-qubit states. Proposition \ref{SLOCCTS} then implies that all states in the W class have tree size $8$.
\end{proof}

What is the tree size of the second most complex class? One can also show in a similar way that the GHZ state cannot be described by the tree $\rm{T_B}$. A reader familiar with entanglement classification would find this obvious because there is no genuine three-qubit entanglement in $\rm{T_B}$. Note that $\rm{T_P}$ is a special case of $\rm{T_B}$ so need not be considered. Therefore, the GHZ state, as well as all the states in its class, has tree size 6.  

Since states are determined only up to a finite precision, it is necessary to consider the more physical definition of the tree size --- the $\epsilon$-approximate tree size. We see in Sec.~\ref{entclas} that adding an arbitrarily small perturbation to $\ket{\rm{W}}$ results in a state in the GHZ class with tree size $6$. Thus, $\rm{TS}_{\epsilon}\left(\ket{\rm{W}}\right)$ is at most $6$ for any positive $\epsilon$. This is also true for all the states in the W class, which means that the maximal tree size of three-qubit states under nonvanishing perturbation is 6.

How large can $\epsilon$ be before $\ket{\rm{W}}$ can be approximated by the tree $\rm{T_B}$ with size $5$? Any normalized state that is described by $\rm{T_B}$ must adopt the following form, up to permutation of qubits:
\beq\label{bisep}
\ket{\rm{T_B}}=\ket{u}\left(\alpha\ket{0}\ket{v}+\beta\ket{1}\ket{v'}\right),
\eeq
where $\ket{u}$, $\ket{v}$, and $\ket{v'}$ are normalized single-qubit states, and $|\alpha|^2+|\beta|^2=1$. This is in general a biseparable state, but it is a product state when $\ket{v}=\ket{v'}$ up to a phase. The $\ket{\mathrm{W}}$ state is particularly easy to work with since it is invariant under permutation of qubits, so we do not need to consider the other two partitions of qubits, $2|13$ and $3|12$, in $\ket{{\rm{T_B}}}$. It is known that the squared overlap $|\braket{\rm{W}|\rm{T_B}}|^2$ obeys the inequality
\beq\label{WBdistance}
|\braket{\rm{W}|\rm{T_B}}|^2\leq \frac{2}{3}.
\eeq
The $2/3$ upper bound is strict because the equality holds when, for instance, $\ket{\rm{T_B}}=\ket{\rm{B}}$ given in Eq.~\eqref{pbgw}. From the definition of the $\epsilon$-approximate tree size and the above inequality one concludes that  $\mathrm{T}_{\epsilon}\left(\ket{\rm{W}}\right)=6$ for $0<\epsilon<1/3$. 

\subsection{Mixed states}
Computing the $\rm{TS}$ of mixed states of two qubits is straightforward. The $\rm{TS}$ of a separable state is 2 and the $\rm{TS}$ of an entangled state, which has at least one entangled pure state in its decomposition, is 4. One can determine whether a mixed state is separable or entangled using the positive partial transpose criterion \cite{Horodecki96}. 

The case of three qubits is more difficult. The $\rm{TS}$ of mixed states of three qubits can be computed based on the entanglement classification of three qubits introduced by Acin \textit{ et. al.} \cite{Acin01}: A mixed state of three qubits belongs to the class $\rm{S}$ of separable states if it can be expressed as a convex sum of separable pure states, class $\rm{B}$ of biseparable states if it can be expressed as a convex sum of separable  and biseparable pure states, class $\rm{W}$ if it can be expressed as a convex sum of separable, biseparable, and $\rm{W}$ pure states, and class $\rm{GHZ}$ if it can be expressed as a convex sum of all possible states. It follows that $\rm{S}\subset \rm{B}\subset \rm{W}\subset \rm{GHZ}$ (see Fig.~\ref{mixclass}). From the definition of the $\rm{TS}$ for mixed states we see that $\rm{TS}=3$ for the set $S$, $\rm{TS}=5$ for the set $\rm{B}\backslash \rm{S}$, $\rm{TS}=8$ for the set $\rm{W}\backslash \rm{B}$, and $\rm{TS}=6$ for the set $\rm{GHZ}\backslash \rm{W}$. 

\begin{figure}[t]
\centering
\includegraphics[scale=0.4]{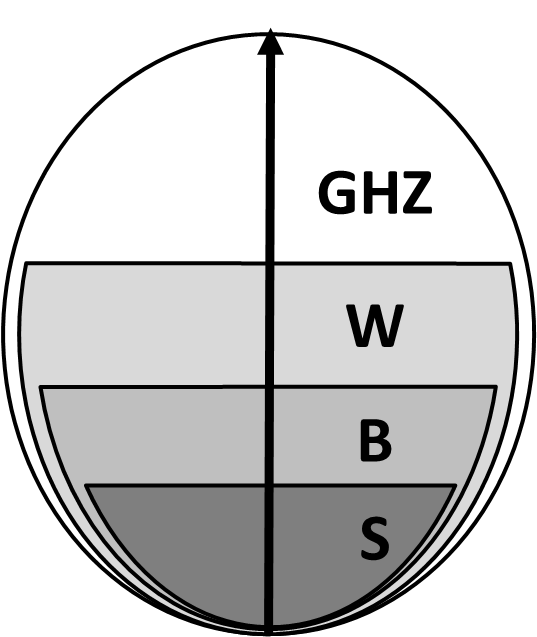}
\caption{Each SLOCC class of  mixed three-qubit  states forms a convex set with  extremal points representing the pure states in the class. The boundaries of smaller sets are infinitesimally close to those of the bigger sets, because for every pure state in the former there is one in the latter that is arbitrarily close to it. The vertical vector shows qualitatively how the generalized Werner state transfoms from one set to the next when the parameter $p$ is increased.}\label{mixclass}
\end{figure}

However, unlike in the case of two qubits, there is no systematic way to determine the entanglement class of an arbitrary mixed state. But it is possible to identify the class of some states with the help of entanglement witnesses \cite{Acin01}. Let us consider the example of the generalized Werner state of three qubits \cite{Eltschka12}
\ba{}{
\rho_{\rm{WS}}(p)=p\ket{\rm{GHZ}}\bra{\rm{GHZ}}+\frac{1-p}{8} \, \rm{I}_8,}
where the parameter $p$ ranges from $0$ to $1$. It is shown that $\rho_{\rm{WS}}(p)$ belongs to the set $\rm{S}$ when $p\leq 1/5$, $\rm{B}\backslash \rm{S}$ when $1/5<p\leq 3/7$, $\rm{W}\backslash \rm{B}$ when $3/7<p\leq p_{W}$, and $\rm{GHZ}\backslash \rm{W}$ when $p_{W}<p\leq 1$ where $p_W\approx 0.695 542 7$ \cite{Eltschka12}. Thus, the generalized Werner state has the maximal $\rm{TS}=8$ when $3/7<p\leq p_{W}$.

\section{Four qubits}\label{4qubit}
In this section we use the same approach as in the previous section to find the most complex four-qubit states. We find that every state can be be described by a tree with at most 16 leaves only if the state is written in the form \eqref{t88}. This form seems to preclude a recursive construction of the most economic description in terms of tree size. Forms that look recursive, like \eqref{Sform}, do require 18 leaves for some states. In the process, we find that the most complex four-qubit state belong to a SLOCC class not described in previous classifications, which we call ``states with irreducible A|BCD form". Finally, as we did for three qubits, we describe the approximate tree size.

\subsection{A\textbar BCD  form}
We begin our search by the observation that any four-qubit state $\ket{\Psi}$ can be written as
\ba{Sform}{
\ket{\Psi}=\ket{0}\ket{\phi_0}+\ket{1}\ket{\phi_1}}
with some three-qubit states $\ket{\phi_0}$ and $\ket{\phi_1}$.  We will refer to it as the \textit{$A|BCD$  form} in the rest of the text. The size of the above decomposition is at most $18$ and reaches the maximal value when both $\ket{\phi_0}$ and $\ket{\phi_1}$ are states in the W class. We now define the \textit{irreducible ${A|BCD}$ form} as follows: 
\begin{definition}
A state $\ket{\Psi}$ has the irreducible $\rm{A|BCD}$  form if for all $\rm{A}\in \{1,2,3,4\}$ and all ILOs $A_1$ the $\rm{A|BCD}$  form of $A_1\ket{\Psi}$ always has $\ket{\phi_0}$ and $\ket{\phi_1}$ in the W class. 
\end{definition}
In other words, both the three-qubit states $\ket{\phi_0}$ and $\ket{\phi_1}$ in the $\rm{A|BCD}$  form cannot be brought out of the W class by switching to a different partition and applying ILOs. Note that permuting qubits in the $\rm{BCD}$ part and applying ILOs on these qubits do not bring $\ket{\phi_0}$ and $\ket{\phi_1}$ out of the $\rm{W}$ class. Therefore, we come to a stronger statement that if a state has the irreducible $\rm{A|BCD}$  form, then for all permutations $\rm{(A,B,C,D)}$ of $(1,2,3,4)$ and ILOs $A_1,A_2,A_3,A_4$ the $\rm{A|BCD}$  form of $A_1\otimes A_2\otimes A_3 \otimes A_4 \ket{\Psi}$ has $\ket{\phi_0}$ and $\ket{\phi_1}$ in the W class. These states were overlooked in the original work on the inductive classification of entanglement for four qubits \cite{Lamata07}.

It is shown in Sec.~\ref{MTS4} that \emph{the states with irreducible $\rm{A|BCD}$  form  are indeed the most complex four-qubit states}, but let us first identify the set of these states. When both $\ket{\phi_0}$ and $\ket{\phi_1}$ are in the W class, first we use ILOs on the last three qubits to transform $\ket{\phi_1}$ to the $\ket{\rm{W}}$ state. We denote the state after this transformation by 
\ba{psi1w}{
\ket{\Psi}=\ket{0}\ket{\rm{\phi_w}}+\ket{1}\ket{\rm{W}},}
where $\ket{\rm{\phi_w}}$ is a state in the W class with the coefficient matrices   
\ba{}{
C_0=\begin{pmatrix}c_{1} &c_2 \\
                   c_3 &c_4 \end{pmatrix},\ \  C_1=\begin{pmatrix}c_{5} &c_6 \\
                   c_7 &c_8 \end{pmatrix}.}
We can assume that $c_2=0$ because if $c_2\neq 0$ we can apply $A_1$ on the first qubit such that $A_1\ket{0}=\ket{0},A_1\ket{1}=-c_2\ket{0}+\ket{1}$. After applying $A_1$ the new state $\ket{\rm{\phi_W}}$ has $c_2=0$. 

By applying a local invertible operator 
\ba{}{
A_1=\begin{pmatrix}
a_{11} &a_{12} \\
a_{21} &a_{22} \end{pmatrix},}
with 
\ba{incond}{
\det(A_1)=a_{11}a_{22}-a_{12}a_{21} \neq 0,}
on the first qubit, we have
\ba{Wmix}{
A_1\ket{\Psi}=&\ket{0}\left(a_{11}\ket{\rm{\phi_w}}+a_{12}\ket{\rm{W}}\right)
\nn \\+&\ket{1}\left(a_{21}\ket{\rm{\phi_w}}+a_{22}\ket{\rm{W}}\right).} Thus, we must find $\ket{\rm{\phi_w}}$ such that  (1) $a_{11}\ket{\rm{\phi_w}}+a_{12}\ket{\rm{W}}$ remains in the W class, and (2) $a_{21}\ket{\rm{\phi_w}}+a_{22}\ket{\rm{W}}$ remains in the W class, for all $a_{ij}$ obeying the invertibility condition. This constraint makes sure that $a_{11}$ and $a_{12}$ are not both zero, and neither are $a_{21}$ and $a_{22}$. If $a_{11}= 0$ then $a_{12}\neq 0$, and thus (1) is always satisfied. If $a_{11}\neq 0$ let $\lambda =a_{12}/a_{11}$ and (1) becomes finding $\ket{\rm{\phi_w}}$ such that $\ket{\rm{\phi_w}}+\lambda\ket{W}$ for all $\lambda$. Similar arguments for the pair  $a_{21}$ and $a_{22}$ result in the same requirement. 

Recall that we can choose $c_2=0$, the coefficient matrices of $\ket{\rm{\phi_w}}+\lambda\ket{\rm{W}}$ is
\ba{}{
C_0=\begin{pmatrix}c_{1} &\lambda \\
                   c_3+\lambda &c_4 \end{pmatrix},\ \ C_1=\begin{pmatrix}c_{5}+\lambda &c_6 \\
                   c_7 &c_8 \end{pmatrix}.}
Now we refer to condition (2) of Proposition~\ref{Lamata}. Condition (2a), which requires $\det{C_0}\neq 0$, cannot be satisfied with all $\lambda$ because $\det{C_0}=-\lambda^2-c_3\lambda+c_1c_4$ has at least one zero regardless of the values of $c_1$, $c_3$, and $c_4$. So we try (2b): We have $\det{C_1}=(c_5+\lambda)c_8-c_6c_7$, which is not zero for all $\lambda$ if and only if $c_8=0$ and $c_6c_7\neq 0$; the equation $\left[\mathrm{tr}(C_1^{-1}C_0)\right]^2=4 \det(C_1^{-1}C_0)$  yields a quadratic equation,  
\ba{}{
a_2\lambda^2+a_1\lambda+a_0=0,}
where $a_0,a_1,a_2$ are functions of $c_1,c_3,\ldots, c_7$. This equation is satisfied for all $\lambda$ if all $a_0,a_1,a_2$ vanish, which yields three constraints:
\ba{}{
(c_6+c_7-c_4)^2&=4c_6c_7,\nn \\
(c_3c_6-c_4c_5)(c_6+c_7-c_4)&=2c_3c_6c_7,\nn \\
(c_3c_6-c_4c_5)^2&=-4c_1c_4c_6c_7.}
And with these constraints, it is not difficult to show that $C_0$ and $C_1$ are linearly independent for all $\lambda$ if and only if $c_4\neq 0$. In Ref.~\cite{Lamata07}, the possibility of $a_0,a_1,a_2$ being all zero was not noticed; as a consequence, the authors missed the family of states with irreducible $\rm{A|BCD}$ form. Indeed, as soon as one of the coefficients is not zero, the quadratic equation has at most two solutions: Then one can choose any $\lambda$ not in the set of solutions to transform at least one of $\ket{\phi_0}$ and $\ket{\phi_1}$ out of the W class.

We need to carry out the same analysis for other partitions $\rm{A}=2,3,4$. The algebra is lengthy but straightforward. As it turns out, all additional equations and inequalities can be derived from the already known constraints. So the state $\ket{\Psi}$ in Eq.~\eqref{psi1w} has the irreducible $\rm{A|BCD}$ form if the coefficients of the three-qubit state $\ket{\rm{\phi_w}}$ obey the following set of constraints:
\ba{}{
&c_4,c_6,c_7\neq 0, \nn \\
&c_2= c_8=0,\nn \\
&(c_6+c_7-c_4)^2=4c_6c_7, \nn \\
&(c_3c_6-c_4c_5)(c_6+c_7-c_4)=2c_3c_6c_7, \nn \\
&(c_3c_6-c_4c_5)^2=-4c_1c_4c_6c_7.}

These constraints can be greatly simplified when we consider two possible cases: $c_1=0$ and $c_1\neq 0$. If $c_1=0$ the above constraints become
\ba{case1}{
&c_4,c_6,c_7 \neq 0, \nn \\
&c_1=c_2=c_3=c_5=c_8=0, \nn \\
&c_4= (\sqrt{c_6}\pm\sqrt{c_7})^2,}
where $\sqrt{z}$ denotes the principal square root of the complex number $z$. If $c_1\neq 0$, by applying $A_1$ on the first qubit such that $A_1\ket{0}=-\ket{0}/c_1, A_1\ket{1}=\ket{1}$ we get a new $\ket{\rm{\phi_w}}$ with $c_1=-1$. Substituting this into the set of equations we obtain
\ba{case2}{
&c_4,c_6,c_7\neq 0, \nn \\
&c_1=-1, \ c_2=c_8=0, \nn \\
&c_4=\left(\frac{c_3}{2}\right)^2,
\ c_6=\left(\frac{c_5}{2}\right)^2,
\ c_7=\left(\frac{c_3-c_5}{2}\right)^2.} 

The simplest example of the first case is $c_4=4,c_6=c_7=1$, which yields the state
\ba{}{
\ket{\Psi}=&\ket{0}\left(4\ket{011}+\ket{101}+\ket{110}\right)\nn \\
+&\ket{1}\left(\ket{001}+\ket{010}+\ket{100}\right).}
It is not difficult to find ILOs to transform $\ket{\Psi}$ to the following more symmetric state:
\ba{stateQ}{
\ket{\Psi^{(4)}}=&\frac{1}{\sqrt{2}}\left(\ket{0}\ket{\rm{W_0}}+\ket{1}\ket{\rm{W_1}}\right),}
where 
\ba{}{
\ket{\rm{W_0}}&=\frac{1}{\sqrt{6}}\left(\ket{110}+\ket{101}-2\ket{011}\right), \nn \\
\ket{\rm{W_1}}&=\frac{1}{\sqrt{6}}\left(\ket{001}+\ket{010}-2\ket{100)}\right).}
Or, explicitly,
\ba{}{\ket{\Psi^{(4)}}=\sqrt{\frac{1}{3}}\bigg[&\frac{1}{2}\left(\ket{0110}+\ket{0101}+\ket{1001}+\ket{1010}\right)\nn \\
&-\ket{0011}-\ket{1100}\bigg].}
Coincidentally, this state was already realized in experiments by using photons from a down-conversion source \cite{Bourennane04,Eibl03}; and the genuine four photon entanglement was confirmed by measuring a witness \cite{Bourennane04}. The state $\ket{\Psi^{(4)}}$ is similar to the Dicke state with two excitations, that is,
\ba{}{
\ket{\rm{D_2}}=&\frac{1}{\sqrt{6}}\big(\ket{0011}+\ket{0101}+\ket{0110}\nn \\
&+\ket{1001}+\ket{1010}+\ket{1100}\big),}
except for the factors of $-2$. Despite this similarity, $\ket{\rm{D_2}}$ does not have irreducible  $\rm{A|BCD}$  form , as one can verify that applying the ILO
\ba{}{
A_1=\begin{pmatrix}
1 &1 \\
1 &-1 \end{pmatrix}}
on the first qubit gives the state $\ket{0}\ket{\rm{GHZ'}}+\ket{1}\ket{\rm{GHZ''}}$ where $\ket{\rm{GHZ'}}$ and $\ket{\rm{GHZ''}}$ belong to the GHZ class.

In the next section we show that a four-qubit state has maximal tree size if and only if it has irreducible $\rm{A|BCD}$ form. Moreover, while the maximal size of the $\rm{A|BCD}$  form is 18, the maximal tree size is only 16, which means that the  $\rm{A|BCD}$ form is not the optimal decomposition for the most complex four-qubit states. 

\subsection{Maximal tree size}\label{MTS4}
First, we list all the four-qubit trees with a $\otimes$ gate at the root as shown in Fig.~\ref{tensor4}. The subscript of each tree indicates the number of its leaves. All the other trees (with a $+$ gate at the root) are combinations of these trees. Then we proceed to prove that the tree size of the states with irreducible $\rm{A|BCD}$ form is 16, which will later be shown to be the maximal tree size of four-qubit states.  
\begin{figure*}
\centering
\includegraphics[scale=0.8]{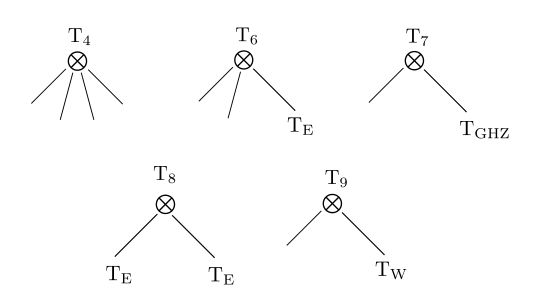}
  \caption{Rooted trees for four qubits with a $\otimes$ gate at the root. The subscript of each tree indicates the number of its leaves. The trees $\rm{T_E}$, $\rm{T_{GHZ}}$, and $\rm{T_{W}}$ are used to label branches.}\label{tensor4}
\end{figure*}
\begin{proposition}\label{4qc1}
A four-qubit state $\ket{\Psi}$ with the irreducible $\rm{A|BCD}$ form has the decomposition
\ba{t88}{
\ket{\Psi}=\ket{\phi_{12}}\ket{\varphi_{34}}+\ket{\phi'_{13}}\ket{\varphi'_{24}},}
where $\ket{\phi}$, $\ket{\varphi}$, $\ket{\phi'}$, $\ket{\varphi'}$ are two-qubit entangled states, and the subscripts indicate the qubits assigned to each partition.
\end{proposition} 
Here $\ket{\phi_{12}}$ is an entangled state of the first qubit and the second qubit and so on. Since $\rm{T_8}$ is a tensor product of two two-qubit entangled states (see Fig.~\ref{tensor4}), it is clear that the above decomposition can be described by the tree $\rm{T_8+T_8}$.  Note that the orders of qubits assigned to each branch are not the same. Hence, this decomposition is not similar to those usually seen in entanglement theory. This ``crossing" of qubits is required to obtain the minimal tree for the states with irreducible $\rm{A|BCD}$ form.

\begin{proof}
We show this by explicit constructions. For the first case where the coefficients are given in Eq.~\eqref{case1}, one can verify that the explicit form is  
\ba{}{
\ket{\phi_{12}}&=(\ket{10}\mp\sqrt{c_6}\sqrt{c_7}\ket{01}),\nn \\
\ket{\varphi_{34}}&=\mp\frac{\sqrt{c_6}}{\sqrt{c_7}}\ket{01}+\ket{10},\nn\\
\ket{\phi'_{13}}&=\sqrt{c_7}(\sqrt{c_7}\pm\sqrt{c_6})\ket{01}+\ket{10},\nn \\
\ket{\varphi'_{24}}&=\frac{\sqrt{c_7}\pm\sqrt{c_6}}{\sqrt{c_7}}\ket{01}+\ket{10}.}
The above state is always well defined because the constraints of Eq.~\eqref{case1} require that $c_7\neq 0$.

For the second case of Eq.~\eqref{case2} the explicit form is more complicated:
\ba{}{
\ket{\phi_{12}}&=\frac{4}{c_5(c_5-c_3)}\ket{10}+\ket{0}\left[\frac{2c_3}{c_5(c_5-c_3)}\ket{0}+\ket{1}\right],\nn \\
\ket{\varphi_{34}}&=\frac{c_5^2}{4}\ket{01}+\left[\frac{c_5}{2}\ket{0}+\frac{c_5(c_5-c_3)}{4}\ket{1}\right]\ket{0},\nn\\
\ket{\phi'_{13}}&=\ket{0}\left[\frac{c_5}{c_3-c_5}\ket{0}+\frac{c_3}{2}\ket{1}\right]+\frac{2}{c_3-c_5}\ket{10},\nn \\
\ket{\varphi'_{24}}&=\frac{c_3-c_5}{2}\ket{10}+\ket{0}\left[\ket{0}+\frac{c_3}{2}\ket{1}\right],}
which is again always well defined since the constraints of Eq.~\eqref{case2} make sure that $c_3 \neq c_5$.
\end{proof}

The decomposition of Eq.~\eqref{t88} turns out to be optimal for all states with irreducible $\rm{A|BCD}$ form. In other words, these states do not possess decompositions with size smaller than 16.
\begin{proposition}\label{4qc2}
If $\ket{\Psi}$ is a state with irreducible $\rm{A|BCD}$  form, its minimal tree is $\rm{T_8+T_8}$. Thus, the tree size of these states is 16. 
\end{proposition}

\begin{proof}
First we need to draw the trees with 15 leaves or less. There are a lot of them, but most are special cases of others. Let us first consider the set of trees shown in Fig.~\ref{tensor4}.  We see that $\rm{T}_4$ is a special case of $\rm{T_6}$ because the two-qubit product state is a special case of $\rm{T_E}$. We denote this relation as $\rm{T_4}\subset \rm{T_6}$. After examining the structures of these trees, one can be convinced that $\rm{T_4}\subset \rm{T_6}\subset \rm{T_7} \subset\rm{T_9}$, $\rm{T_6}\subset\rm{T_8}$ and $\rm{T_7}\subset \rm{T_4}+\rm{T_4}$. From these relations we have $\rm{T_6+T_6}\subset \rm{T_7+T_7}\subset \rm{T_4+T_4+T_7}$ and so on. After listing the trees with at most 15 leaves we see that all of them are special cases of the set of trees with exactly 15 leaves. This set is shown in Fig.~\ref{F415}. 
\begin{figure}
\centering
\includegraphics[scale=0.8]{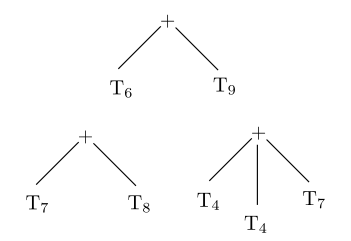}
  \caption{Rooted trees of four qubits with 15 leaves. All the trees with fewer than 15 leaves are special cases of these trees.}\label{F415}
\end{figure}

We now prove Proposition~\ref{4qc2} by showing that if a state is described by a tree with at most 15 leaves, it cannot have the irreducible $\rm{A|BCD}$  form. Only $\rm{T_6+T_9}$ is considered as the arguments for the other trees are similar. For better clarity, we draw $\rm{T_6+T_9}$ explicitly in Fig.~\ref{T15}.
\begin{figure*}
\centering
\includegraphics[scale=0.6]{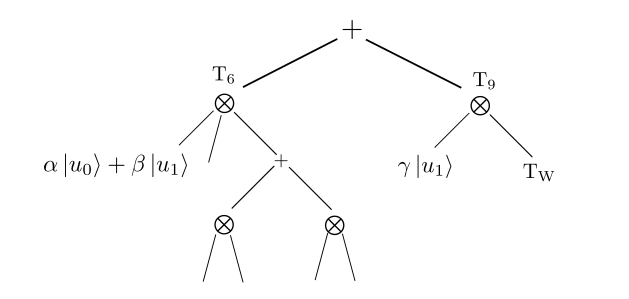}
  \caption{Explicit drawing of $\rm{T_6+T_9}$ with its two main branches $\rm{T_6}$ and $\rm{T_9}$.}\label{T15}
\end{figure*}
Let us denote by $\gamma\ket{u_1}$ the single-qubit state assigned to the leaf at the root of $\rm{T_9}$ with $\ket{u_1}$ normalized. The same qubit may be assigned to any leaf on the $\rm{T_6}$ branch. There are two inequivalent cases. In the first, this qubit is assigned to one of the two leaves at the root of $\rm{T_6}$. Its state can be then expressed as $\alpha \ket{u_0}+\beta\ket{u_1}$ where $\ket{u_0}$ is the normalized state that forms a basis with $\ket{u_1}$. The four-qubit state described by $\rm{T_{6}+T_{9}}$ is
\ba{rS}{
\ket{\varphi}&=\left(\alpha \ket{u_0}+\beta\ket{u_1}\right)\ket{\rm{T_B}}+\gamma\ket{u_1}\ket{\rm{T_W}}\nn \\
&=\alpha \ket{u_0}\ket{\rm{T_B}}+\ket{u_1}\left(\beta\ket{\rm{T_B}}+\gamma\ket{\rm{T_W}}\right).}
Note that the tree states with letters as subscripts are three-qubit states (see Fig.~\ref{F38}). The size of the three-qubit state $\beta\ket{\rm{T_B}}+\gamma\ket{\rm{T_W}}$ is at most 8, so the size of the above $\rm{A|BCD}$  form is at most 15 and hence cannot be maximal. 

In the second case, the qubit assigned to the leaf at the root of $\rm{T_9}$ is assigned to one of the leaves in the two-qubit entangled subtree $\rm{T_E}$ of $\rm{T_6}$. Then, we can express the two-qubit state of $\rm{T_E}$ as $\alpha\ket{u_0}\ket{v_0}+\beta\ket{u_1}\ket{v_1}$. After writing down $\ket{\varphi}$ and grouping the terms with $\ket{u_1}$ together, one sees that it has the form of Eq.~\eqref{rS} with $\ket{\rm{T_P}}$ instead of $\ket{\rm{T_B}}$. The size of this $\rm{A|BCD}$ form is at most 13 and is again not maximal.

For the $\rm{T_7+T_8}$ tree, denote by $\gamma\ket{u_1}$ the single-qubit state assigned to the leaf at the root of $\rm{T_7}$. Using the same procedure one can show that a state described by this tree has an $\rm{A|BCD}$  form with size at most 15. Similarly, a state described by $\rm{T_4+T_4+T_7}$ has an $\rm{A|BCD}$  form with size at most 16. Recall that the maximal size of the $\rm{A|BCD}$ form is 18; all the trees with 15 leaves cannot describe a state that has irreducible $\rm{A|BCD}$ form. Since the trees with fewer than 15 leaves are special cases of the trees with 15 leaves, we conclude that all the trees with 15 leaves or fewer cannot describe the states with irreducible $\rm{A|BCD}$ form. Therefore, the $\rm{T_8+T_8}$ tree is the optimal decomposition for these states and their tree size is 16.
\end{proof}

Now the tree size of the states with irreducible $\rm{A|BCD}$ form has been found, one still needs to prove that these states are the most complex, that is, the other states (with no irreducible $\rm{A|BCD}$ form) have smaller tree size.
\begin{proposition}\label{4qc3}
If a four-qubit state $\ket{\Psi}$ does not have an irreducible $\rm{A|BCD}$ form, it can be described by a tree with at most 15 leaves.
\end{proposition}

\begin{proof}
If $\ket{\Psi}$ does not have the irreducible $\rm{A|BCD}$ form, there exists a partition $\rm{A|BCD}$ and an ILO $A_1$ such that after the application of this ILO we have
\ba{}{\ket{\Psi}=\ket{0}\ket{\phi_0}+\ket{1}\ket{\phi_1},}
where at least one of the three-qubit states, say $\ket{\phi_1}$, is not in the $\rm{W}$ class. If $\ket{\phi_1}$ is a biseparable state, it is clear that $\ket{\Psi}$ can be described by the tree $\rm{T_9+T_6}$ with 15 leaves. If $\ket{\phi_1}$ is in the GHZ class, we use ILOs on the last three qubits to transform $\ket{\Psi}$ to
\ba{}{
\ket{\Psi}=\ket{0}\ket{\rm{\phi_w}}+\ket{1}\ket{\rm{GHZ}}.}
Consider the state $\ket{\rm{\phi_w}}+\lambda\ket{\rm{GHZ}}$, a necessary condition for this state to be in the W class is that $C_0C_1^{-1}$ or $C_1^{-1}C_0$ has only one eigenvalue. Both cases yield the same equation
\ba{}{
\lambda^4+a_3\lambda^3+a_2\lambda^2+a_1\lambda+a_0=0,}
where $a_0,a_1,a_2,a_3$ are functions of the coefficients of $\rm{\phi_w}$. This equation has at most four distinct solutions. Thus, it is always possible to find a value $\lambda^{\ast}\neq 0$ that is not in the set of solutions. Then, $\ket{\rm{\phi_w}}+\lambda^{\ast}\ket{\rm{GHZ}}$ is not in the W class and hence can be described by $\rm{T_{GHZ}}$. Next, we apply on the first qubit an ILO such that $A_1\ket{0}=\ket{0},A_1\ket{1}=\lambda^{\ast}\ket{0}+\ket{1}$, which is invertible since $\lambda^{\ast}\neq 0$. After that the state becomes
\ba{}{
  \ket{\Psi}=\ket{0}(\ket{\rm{\phi_w}}+\lambda^{\ast}\ket{\rm{GHZ}})+\ket{1}\ket{\rm{GHZ}},}
which can be described by $\rm{T_7+T_7}$ with 14 leaves.
\end{proof}

A direct corollary of Propositions \ref{4qc2} and \ref{4qc3} is that \textit{the states with irreducible $A|BCD$ form are the states with the maximal tree size, and vice versa}. The maximal tree size of four-qubit states is therefore 16. Note that for $2\leq n \leq 4$ \emph{the maximal $TS$ is $2^n$, which is the dimension of the Hilbert space}. Whether this relation holds for all $n\geq 2$ remains an open question.

\subsection{Approximate tree size}
What is the maximal $\epsilon$-approximate tree size of four qubit states? Even in the worst case scenario when $\ket{\phi_0}$ and $\ket{\phi_1}$ in Eq.~\eqref{Sform} are in the W class, we know from Sec.~\ref{entclas} that they can be approximated with arbitrary precision by two states in the GHZ class. More concretely, for $\epsilon$ arbitrarily close to $0$ there exists a state of the form
\ba{}{
\ket{\varphi}= \ket{0}\ket{\rm{GHZ'}}+\ket{1}\ket{\rm{GHZ''}}} 
such that $|\braket{\Psi|\varphi}|^2\geq 1-\epsilon$. Here $\ket{\rm{GHZ'}}$ and $\ket{\rm{GHZ''}}$ are the two states in the GHZ class. Because the tree size of $\ket{\varphi}$ is at most 14, we conclude that the $\rm{TS}_{\epsilon}\left(\ket{\Psi}\right)$ is at most 14 for every $\epsilon >0$. Thus, if fluctuation over a neighborhood is allowed, the maximal tree size of four-qubit states is at most 14. 

\begin{proposition}\label{cpl4}
The $\epsilon$-approximate tree size of $\ket{\Psi^{(4)}}$ is 14 for $0<\epsilon<\frac{1}{12}$.
\end{proposition}
\begin{proof}
We need to show that if $\ket{\varphi}$ is a state described by a tree with fewer than 14 leaves, then $|\braket{\varphi|\Psi^{(4)}}|^2\leq 1-1/12=11/12$. In other words, the states with smaller size are a finite distance away from $\ket{\Psi^{(4)}}$. For this purpose we employ the same elimination procedure used in Sec.~\ref{3qubit} to find the most complex three-qubit state. First, we draw all the trees with 13 leaves or fewer. We observe that all of the trees in this set are special cases of the four particular trees listed in Fig.~\ref{F413}. Thus, Proposition~\ref{cpl4} holds if we can show that the states described by these four trees are a finite distance away from $\ket{\Psi^{(4)}}$. 
\begin{figure}
\centering
\includegraphics[scale=0.5]{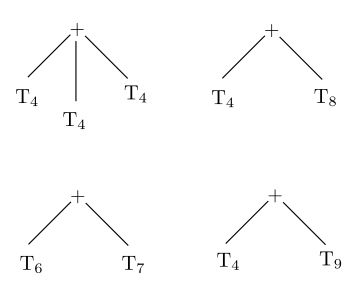}
  \caption{Rooted trees of four qubits with sizes 12 and 13.}\label{F413}
\end{figure}

\textit{Eliminating $\rm{T_4+T_8}$, $\rm{T_6+T_7}$, and $\rm{T_4+T_9}$---} Using the argument in the proof of Proposition \ref{4qc3} that leads to Eq.~\eqref{rS}, one can show that a state described by one of $\rm{T_4+T_8}$, $\rm{T_6+T_7}$, and $\rm{T_4+T_9}$ has the form 
\ba{4qs}{
\ket{\varphi}=\alpha\ket{u_0}\ket{\rm{T_B}}+\beta\ket{u_1}\ket{\phi},}
where $|\alpha|^2+|\beta|^2=1$, $\ket{u_0}$ and $\ket{u_1}$ are two orthonormal single-qubit states, $\ket{\phi}$ is a three-qubit state, and $\ket{\rm{T_B}}$ is a biseparable state (product states are treated as a special case of biseparable states). In addition, $\ket{\phi}$ and $\ket{\rm{T_B}}$ are both normalized. One sees later that it is the biseparable component $\ket{\rm{T_B}}$ of $\ket{\varphi}$ that keeps it away from $\ket{\Psi^{(4)}}$. 

The state $\ket{\varphi}$ has the  form of the bipartite cut $\rm{A|BCD}$ where any of the four qubits can be assigned to part A. We first consider the case when they are the states of the first qubit. Complications due to permutation of qubits will be dealt with later. As in the case for three qubits, we look at the overlap
\ba{qphi}{
|\braket{\Psi^{(4)}|\varphi}|^2 &=|\alpha\braket{\Psi^{(4)}|u_0\rm{T_B}}+\beta\braket{\Psi^{(4)}| u_1\phi}|^2\nn \\ 
&\leq |\braket{\Psi^{(4)}|u_0\rm{T_B}}|^2+|\braket{\Psi^{(4)}|u_1\phi}|^2,}
where the last line follows from an application of Cauchy-Schwarz inequality. Moreover, we have
\ba{qu0}{
|\braket{\Psi^{(4)}|u_0\rm{T_B}}|^2&=\frac{1}{2}|\braket{0|u_0}\braket{\rm{W_0}|\rm{T_B}}+\braket{1|u_0}\braket{\rm{W_1}|\rm{T_B}}|^2\nn \\
&\leq \frac{1}{2}\left(|\braket{\rm{W_0}|\rm{T_B}}|^2+|\braket{\rm{W_1}|\rm{T_B}}|^2\right),}
and
\ba{qu1}{
|\braket{\Psi^{(4)}|u_1\phi}|^2&=\frac{1}{2}|\braket{0|u_1}\braket{\rm{W_0}|\phi}+\braket{1|u_1}\braket{\rm{W_1}|\phi}|^2\nn \\
&\leq \frac{1}{2}\left(|\braket{\rm{W_0}|\phi}|^2+|\braket{\rm{W_1}|\phi}|^2\right).}

Since $\ket{\rm{W_0}}$ and $\ket{\rm{W_1}}$ are two orthonormal states, we can expand $\ket{\phi}$ in a basis in which $\ket{\rm{W_0}}$ and $\ket{\rm{W_1}}$ are basis vectors. From this observation we can conclude immediately that $|\braket{\rm{W_0}|\phi}|^2+|\braket{\rm{W_1}|\phi}|^2\leq 1$. This and Eqs.~\eqref{qphi}, \eqref{qu0}, and \eqref{qu1} yields
\ba{ineq}{
|\braket{\Psi^{(4)}|\varphi}|^2\leq \frac{1}{2}\left(|\braket{\rm{W_0}|\rm{T_B}}|^2+|\braket{\rm{W_1}|\rm{T_B}}|^2+1\right).}
Thus, we need to maximize 
\ba{}{
f=|\braket{\rm{W_0}|\rm{T_B}}|^2+|\braket{\rm{W_1}|\rm{T_B}}|^2}
to find the largest overlap between $\ket{\Psi^{(4)}}$ and $\ket{\varphi}$.

As a biseparable state $\ket{\rm{T_B}}$ has the following form:
\ba{}{
\ket{\rm{T_B}}=&\left(a\ket{0}+b\ket{1}\right) \nn \\
&\otimes\left(c_{00}\ket{00}+c_{01}\ket{01}+c_{10}\ket{10}+c_{11}\ket{11}\right)
}
with the constraints
\ba{constr}{
|a|^2+|b|^2=
|c_{00}|^2+|c_{01}|^2+|c_{10}^2+|c_{11}|^2&=1.} Substituting this in $f$ we obtain
\ba{}{
f=\frac{1}{6}\left(|b(c_{01}+c_{10})-2a c_{11}|^2+|a(c_{01}+c_{10})-2bc_{00}|^2\right).}
Maximizing $f$ with respect to the constraints of Eq.~\eqref{constr} gives us $f_{\rm{max}}=2/3$. 

Let us now consider the situations when instead of qubit 1, qubits 2,3,4 are assigned to part A of the bipartite cut $\rm{A|BCD}$. Since we are concerned with the overlap $|\braket{\Psi^{(4)}|\varphi}|$, this is equivalent to keeping $\ket{\varphi}$ unchanged while permuting the qubits in $\ket{\Psi^{(4)}}$. Under permutation $\ket{\Psi^{(4)}}$ still has the same form as Eq.~\eqref{stateQ} but the factor 2 changes its place within $\ket{\rm{W_0}}$ and $\ket{\rm{W_1}}$. Thanks to the high symmetry in the form of $\ket{\Psi^{(4)}}$, there are only two different cases:
\ba{}{
\ket{\Psi^{(4)}_1}=&\frac{1}{\sqrt{12}}\big[\ket{0}\left(\ket{110}-2\ket{101}+\ket{011}\right)\nn \\
&+\ket{1}\left(\ket{001}-2\ket{010}+\ket{100}\right)\big],\nn \\
\ket{\Psi^{(4)}_2}=&\frac{1}{\sqrt{12}}\big[\ket{0}\left(-2\ket{110}+\ket{101}+\ket{011}\right)\nn \\
&+\ket{1}\left(-2\ket{001}+\ket{010}+\ket{100}\right)\big].}
Following the same argument as for $\ket{\Psi^{(4)}}$, we arrive at the same inequality as in Eq.~\eqref{ineq} for $\ket{\Psi^{(4)}_1}$ and $\ket{\Psi^{(4)}_2}$ with
\beq
f_1=\frac{1}{6}\left[|a c_{11}+b(c_{10}-2c_{01})|^2+|bc_{00}+a(c_{01}-2c_{10})|^2\right]
\eeq
for $\ket{\Psi^{(4)}_1}$, and
\beq
f_2=\frac{1}{6}\left[|a c_{11}+b(c_{01}-2c_{10})|^2+|bc_{00}+a(c_{10}-2c_{01})|^2\right]
\eeq
for $\ket{\Psi^{(4)}_2}$. Since $f_1$ becomes $f_2$ after interchanging $c_{01}\leftrightarrow c_{10}$, they have the same maximum. And maximizing $f_1$ with respect to the constraints gives $f_{1\rm{max}}= 5/6$. Since $f_{1\rm{max}}>f_{\rm{max}}$, we have, for all permutations of qubits,
\beq\label{max1}
|\braket{\Psi^{(4)}|\varphi}|^2\leq \frac{1+f_{1\rm{max}}}{2}=\frac{11}{12}.
\eeq

\textit{Eliminating $\rm{T_4+T_4+T_4}$---} The final step is to eliminate the last tree, $\rm{T_4+T_4+T_4}$, which does not adopt the convenient form of Eq.~\eqref{4qs}. Labeling the leaves of $\rm{T_4+T_4+T_4}$ by $x_{i}\ket{0}+x'_{i}\ket{1}, i=1,\ldots, 4$ for the first branch, $y_{i}\ket{0}+y'_{i}\ket{1}, i=1,\ldots, 4$ for the second, and $z_{i}\ket{0}+z'_{i}\ket{1}, i=1,\ldots, 4$ for the final one, we write down a state described by $\rm{T_4+T_4+T_4}$ as
\ba{}{
\ket{\varphi}=&\bigotimes_{i=1}^{4}\left(x_{i}\ket{0}+x'_{i}\ket{1}\right)\nn \\
+&\bigotimes_{i=1}^{4}\left(y_{i}\ket{0}+y'_{i}\ket{1}\right) \nn \\
+&\bigotimes_{i=1}^{4}\left(z_{i}\ket{0}+z'_{i}\ket{1}\right),
}
and find the maximal value of $|\braket{\Psi^{(4)}|\varphi}|^2$ subjected to the constraint $|\braket{\varphi|\varphi}|=1$. Numerical optimization gives  $|\braket{\Psi^{(4)}|\varphi}|^2\leq 8/9 <11/12$. Comparing this with the inequality of Eq.~\eqref{max1}, we conclude that $|\braket{\Psi^{(4)}|\varphi}|^2\leq 11/12$ for all states $\ket{\varphi}$ described by a tree with at most $13$ leaves, and Proposition \ref{cpl4} follows immediately from this inequality.
\end{proof}
\subsection{A witness for maximal tree size}
The result of the previous section helps us construct the following witness to detect the states with maximal tree size of four qubits
\ba{}{
\mathbb{W}=\frac{11}{12}\rm{I} - \ket{\Psi^{(4)}}\bra{\Psi^{(4)}}.}
For a given state $\ket{\varphi}$, if $\braket{\varphi|\mathbb{W}|\varphi}<0$ then $|\braket{\varphi|\Psi^{(4)}}|^2>11/12$ and therefore the tree size of $\ket{\varphi}$ must be 14. Thus, when the average value of $\mathbb{W}$ is negative we know that the state has the maximal $\epsilon-$approximate tree size.

In the experiment described in Ref.~\cite{Bourennane04} the state $\ket{\Psi^{(4)}}$ was created and its multi-particle entanglement was confirmed with the witness
\ba{}{
\mathbb{W'}=\frac{3}{4}\rm{I}-\ket{\Psi^{(4)}}\bra{\Psi^{(4)}},}
which is then broken down to a sum of locally measurable operators. We have
\ba{}{
\mathbb{W}=\frac{1}{6}\rm{I}+\mathbb{W'}.}
From the experimental data the authors obtain $\braket{\mathbb{W'}}=-0.151 \pm 0.01$, which yields $\braket{\mathbb{W}}=0.02 \pm 0.01$. Thus, the state created in this experiment does not lead to $\braket{\mathbb{W}}<0$. In other words, the fidelity is not high enough to confirm maximal tree size.

\section{Conclusion}
In this paper we develop a procedure for computing the tree size of a state when the number of qubits is 3 and 4. The states with maximal tree size are identified; and it is shown that these states form a set of zero measure. The family of four-qubit states with maximal tree size is an entanglement class not described in the previous works on the inductive method of entanglement classification. The calculation is extended to mixed states of three qubits and an example of a mixed state with maximal tree size is given. Since our method of finding the minimal tree and tree size is based on an exhaustive elimination of smaller trees, it quickly becomes intractable as the number of qubits increases. Numerical investigation is probably needed if one hopes to find the states with maximal tree size for more than four qubits.

\begin{acknowledgements}
This work is supported by the Centre for Quantum Technologies (CQT). CQT is a Research Centre of Excellence funded by the Ministry of Education and National Research Foundation of Singapore. R. R. is supported by the Brazilian agency CNPq by means of the Science Without Borders program.
\end{acknowledgements}

\end{document}